\documentclass[aps, pra,preprint,groupedaddress]{revtex4-1}

\usepackage{graphicx}
\usepackage{dcolumn}
\usepackage{bm}
\usepackage{amssymb}
\usepackage{amsmath}
\usepackage{amsthm}
\usepackage{color}
\usepackage{epstopdf}
\usepackage{enumerate}

\usepackage[twoside,
left=20mm,
right=20mm,
top=40mm,
]{geometry}

\newtheorem{theorem}{Theorem}
\newtheorem{lemma}{Lemma}

\begin{document}

\title{Entanglement as upper bounded for the nonlocality of a general two-qubit system}

\author{Zhaofeng Su}
\email{zfsu@ustc.edu.cn}

\author{Haisheng Tan}

\author{Xiangyang Li}
\affiliation{LINKE Lab, School of Computer Science and Technology, University of Science and Technology of China, 443 Huangshan Road, Shushan District, Hefei City, Anhui Province 230027, China.}

\date{\today}

\begin{abstract}
  Nonlocality and entanglement are not only the fundamental characteristics of quantum mechanics but also important resources for quantum information and computation applications. Exploiting the quantitative relationship between the two different resources is of both theoretical and practical significance. The common choice for quantifying the nonlocality of a two-qubit state is the maximal violation of the Clauser-Horne-Shimony-Holt inequality. That for entanglement is entanglement of formation, which is a function of the concurrence. In this paper, we systematically investigate the quantitative relationship between the entanglement and nonlocality of a general two-qubit system. We rederive a known upper bound on the nonlocality of a general two-qubit state, which depends on the state's entanglement. We investigate the condition that the nonlocality of two different two-qubit states can be optimally stimulated by the same nonlocality test setting and find the class of two-qubit state pairs that have this property. Finally, we obtain the necessary and sufficient condition that the upper bound can be reached.

\end{abstract}


\maketitle

\section{Introduction}
Entanglement and nonlocality are two of the most fundamental characteristics of quantum mechanics~\cite{RH09, NB14}. Besides the fundamental significance for quantum mechanics, they are also indispensable resources for quantum information which is a discipline that has blossomed in the past three decades. Entanglement and nonlocality are core inherent reasons that  quantum information processing tasks have great advantages over their classical counterparts~\cite{NL00}. In recent years, many applications of entanglement and nonlocality have been proposed, which include communication complexity~\cite{CB10}, quantum cryptography~\cite{BB92}, randomness generation~\cite{DG13}, quantum repeaters~\cite{HW98,Zhaofeng18PRA}, and device-independent quantum computation~\cite{BH05}.

Entanglement is a ``spooky'' feature of quantum mechanics, which was first recognized by Einstein, Podolsky, and Rosen (EPR) eight decades ago~\cite{AE35}. There exist quantum states of a composite quantum system which cannot be interpreted as ensembles of product states. This feature is known as entanglement. A quantum state which has this feature is said to be entangled. Otherwise, it is a separable state. Mathematically, the state $\rho$ of a bipartite quantum system $A\otimes B$ is separable if it can be decomposed into the following form:
\begin{align}
   \rho  = \sum_{k} p_{k} \rho_{k}^{A} \otimes \rho_{k}^{B}
\end{align}
where $p_{k} \ge 0$ and $\sum_{k} p_{k} = 1$, and $\rho_{k}^{A}$ and $\rho_{k}^{B}$ are density operators of the corresponding subsystems $A$ and $B$, respectively. Otherwise, the state $\rho$ is entangled.
Entanglement of a bipartite quantum state is quantified by entanglement of formation~\cite{CB96}. In the following discussion, we refer to ``entanglement of formation'' simply as ``entanglement''.

Nonlocality is another fundamental characteristic of quantum mechanics, which is based on the correlation of observations of separated objects.
Recall that locality and realism were considered two of the most basic principles in classical physics. According to local realism theory, there is a complete description, which is known as a hidden variable, between two separated objects. Once the hidden variable is specified, the observations of the two objects are independent. A correlation $P(ab|xy)$ generated from two separated observers that satisfies the principles should admit a local hidden variable model (LHVM)
\begin{align}\label{eq:BiLHVM}
   P(ab|xy) = \sum_{\lambda} q_{\lambda}P_{\lambda}(a|x)P_{\lambda}(b|y),
\end{align}
where $q_{\lambda}$ is the probability of the hidden variable being $\lambda$, and $P_{\lambda}(a|x)$ and $P_{\lambda}(b|y)$ are marginal probabilities. Any correlation that satisfies the LHVM in Eq.~(\ref{eq:BiLHVM}) is said to be local correlation. On the contrary, a correlation is said to be nonlocal if it does not admit any LHVM.
In 1964, Bell showed that the predictions of quantum theory are incompatible with those of any physical theory that satisfies the local realism theory~\cite{Bell64}. This phenomenon is known as quantum nonlocality. Namely, a quantum system can generate nonlocal correlations. Any quantum state that can generate nonlocal correlations is said to be a nonlocal state, which plays a indispensable role in many quantum information processing applications.

However, it is not straightforward to determine whether a quantum state is local or nonlocal via the definition in Eq.~(\ref{eq:BiLHVM}). A clever way of detecting nonlocality is via the violation of a type of inequalities, which are known as Bell inequalities. Bell inequalities give upper bounds on all local quantum states. Thus, the violation of any Bell inequality is a sufficient evidence for nonlocality. The most popular Bell inequality for the bipartite LHVM in Eq.~(\ref{eq:BiLHVM}) is the Clauser-Horne-Shimony-Holt (CHSH) inequality. The CHSH inequality is maximally violated by the maximally entangled state $(|00\rangle + |11\rangle)/\sqrt{2}$ with the maximal violation being $2\sqrt{2}$~\cite{JF69}.

Moreover, a common choice for quantifying nonlocality is through the amount of maximal violation of a Bell inequality~\cite{NB14}. In 1995, Horodecki \textit{et al}. developed a complete characterization for the violation of the CHSH inequality by the arbitrary state of a two-qubit system~\cite{RH95}. Bipartite nonlocal resources are also important for generating multipartite nonlocality~\cite{ZS17}.

The relationship between quantum entanglement and nonlocality is always of great research interest. Although quantum entanglement and nonlocality are different resources~\cite{NB05}, they are closely related. It is obvious that all separable states are local. Thus, being entangled is a necessary condition of being nonlocal. Any purely entangled two-qubit state violates the CHSH inequality~\cite{GP92}, while it is not true for mixed states. In 1989, Werner found a class of mixed states which are entangled while admiting LHVMs~\cite{RF89}. Werner's breakthrough convinces us that the existence of entanglement is not a sufficient condition for the existence of nonlocality. A large number of studies have been reported for researching on the qualitative relationship between quantum entanglement and nonlocality since then. On the contrary, the quantitative relationship between quantum entanglement and nonlocality is not fully understood yet.

In the past two decades, partial results have been reported on the quantitative relationship.
In 2002, Verstraete and Wolf pointed out that the nonlocality of a general two-qubit state is upper bounded by its entanglement~\cite{FV02}. They further found a class of two-level density matrices that can reach the upper bound, which would turn out to be a complete description for reaching the upper bound in our current paper.
In 2011, Batle and Casas explored the quantitative relationship for the general two-qubit pure states expressed in the Bell basis~\cite{BC11}.
Bartkiewicz \textit{et al}. investigated the two-qubit states that have extremal entanglement with respect to different entanglement measures for a given CHSH violation~\cite{KB13}. In the case of concurrence, Bartkiewicz \textit{et al}. found that the low bound of entanglement for a given CHSH violation can be achieved by the pure states and Bell diagonal states. In other words, the pure states and Bell diagonal states can reach the the upper bound of the CHSH violation for a given concurrence.

In the known literature, only a few examples are mentioned to reach the upper bound of Verstraete and Wolf. It is an open problem to find the complete set of states for reaching the upper bound  when the general two-qubit states are considered.

In this paper, we exploit the necessary and sufficient condition for equality in the quantitative relationship between the entanglement and nonlocality of a general two-qubit system. We find that the example in Verstraete and Wolf's paper~\cite{FV02} turns out to be the complete set of states for reaching the upper bound.  Our paper is organized as follows. In Sec.~\ref{sec:MEntanglementandNL}, we give a brief introduction to the measure of quantum entanglement and nonlocality, respectively. As the prerequisites for solving our main question, we exploit the condition that two different general two-qubit states have the same optimal nonlocality test setting in Sec.~\ref{sec:conditionSameOptimalCHSHOperator} and figure out the class of two-qubit pure states that have the same optimal nonlocality test setting in Sec.~\ref{sec:2qubitSameOptimalCHSHOperator}. In Sec.~\ref{sec:NLupperbound}, we rederive Verstraete and Wolf's upper bound on the nonlocality of a general two-qubit state, which depends on the state's entanglement. Finally, we prove the necessary and sufficient condition for reaching the upper bound .

\section{The measure of entanglement and nonlocality}\label{sec:MEntanglementandNL}
As quantum entanglement and nonlocality are indispensable resources for quantum information processing tasks, quantifying the entanglement and nonlocality of a quantum system is of both theoretical and practical significance. In this section, we give detailed descriptions of the quantitative measure of quantum entanglement and nonlocality of arbitrary two-qubit states, respectively. We also find out the correct quantum measurement settings for stimulating the most nonlocality from the quantum state.

\subsection{Entanglement of two-qubit state}
The entanglement $E(\psi)$ of a bipartite pure state $|\psi\rangle$ is the asymptotic number of standard singlets required to locally prepare the state. It also equals the von Neumann entropy of either of the two subsystems~\cite{CH96}. That is,
\begin{align}
   E(\psi) = -{\rm tr}(\rho_{A}\log_{2} \rho_{A}) = -{\rm tr}(\rho_{B}\log_{2} \rho_{B}),
\end{align}
where $\rho_{A} = {\rm tr}_{B}(|\psi\rangle \langle \psi|)$ and $\rho_{B} = {\rm tr}_{A} (|\psi\rangle \langle \psi|)$. The entanglement of a mixed state $\rho$ is defined as the least expected entanglement of any ensemble of pure states of $\rho$. That is,
\begin{align}
   E(\rho) = min \sum_{i}p_{i} E(\psi_{i}),
\end{align}
where the minimization is over all possible pure decompositions of $\rho$ such that $\rho = \sum_{i}p_{i}|\psi_{i}\rangle \langle \psi_{i}|$.

Hill and Wootters showed that the entanglement of a two-qubit pure state is closely related to the concurrence of the state~\cite{SW97}. The concurrence of a two-qubit pure state $|\psi\rangle$ is defined as
\begin{align}\label{eq:concurrence2qubitpure}
   C(\psi) = |\langle \psi| \tilde{\psi}\rangle|,
\end{align}
where $| \tilde{\psi}\rangle = (Y \otimes Y)| \psi^{*}\rangle$. Then, the entanglement of the pure state $|\psi\rangle$ can be written as
\begin{align}\label{eq:eof2qubit}
   E(\psi) = h(\frac{1+ \sqrt{1 - C(\psi)^{2}}}{2}),
\end{align}
where $h(x)$ is known as Shannon's entropy function. The function is defined as $h(x) \equiv -x \log_{2} x - (1-x)\log_{2} (1-x)$ with $x\in [0,1]$. To keep the completeness of the definition, it denotes that $\log_{2} 0 = 0$.

Further, Wootters showed that the formula for the entanglement of a two-qubit pure state in Eq.~(\ref{eq:eof2qubit}) also holds true for mixed states~\cite{WW98}. He derived an analytical expression for the concurrence of a two-qubit state $\rho$ as follows:
\begin{align}
   C(\rho) = max\{0, \lambda_{1} - \lambda_{2} - \lambda_{3} - \lambda_{4} \},
\end{align}
where $\lambda_{1} \ge \lambda_{2} \ge \lambda_{3} \ge \lambda_{4}$ are the eigenvalues of the Hermitian matrix $\sqrt{\sqrt{\rho}\tilde{\rho}\sqrt{\rho}}$. Here, the operator $\tilde{\rho} \equiv (Y\otimes Y)\rho^{*}(Y\otimes Y)$ is the spin flip of $\rho$. Namely, the relationship of concurrence and entanglement is
\begin{align}
   E(\rho) = h(\frac{1+ \sqrt{1 - C(\rho)^{2}}}{2})
\end{align}
for a general two-qubit state.

The key process of evidencing this result is that Wootters showed the existence of a special decomposition of a general two-qubit state. We conclude the fact in Lemma \ref{lemma:decomposition4concurrence}.
\begin{lemma}\label{lemma:decomposition4concurrence}
   Suppose $\rho$ is an arbitrary density operator of a two-qubit system. There exits a pure decomposition $\rho = \sum_{k} p_{k} |\psi_{k}\rangle\langle \psi_{k}|$ with $p_{k} \ge 0$ and $\sum_{k}p_{k} = 1$ such that the concurrence $C(\psi_{k}) = C(\rho)$ for each state $|\psi_{k}\rangle$.
\end{lemma}
Note that the concurrence of a density operator is a quantity ranging from 0 to 1. The separable states correspond to concurrence 0 and the maximally entangled states correspond to concurrence 1. The entanglement $E(\rho)$ is monotonically increasing as $C(\rho)$ goes from 0 to 1. Thus, concurrence can act as a measure of entanglement in its own right.

\subsection{Nonlocality of two-qubit state}
The simplest scenario for generating quantum correlations consists of two separated parties, say,  Alice and Bob. A two-qubit quantum state $\rho$ is shared between them. Each of the two parties performs a dichotomic measurement which is taken from two possible choices. Suppose Alice performs the measurement $x$ and observes an outcome $a$. Those of Bob are $y$ and $b$, respectively. In this paper, we only consider projective measurements on single-qubit systems where the outcomes $a,b\in\{+1, -1\}$. Then, a quantum correlation can be generated from the scenario as follows:
\begin{align}\label{eq:BiQuantumCorrelation}
   P(ab|xy) = {\rm tr}(\rho M^{x}_{a} \otimes M^{y}_{b} ),
\end{align}
where $M^{x}_{a}$ is the measurement operator of Alice when she chooses measurement $x$ and obtains outcome $a$ and $M^{y}_{b}$ is the corresponding measurement operator of Bob. If there exist measurement settings for Alice and Bob such that the generated quantum correlation in Eq.~(\ref{eq:BiQuantumCorrelation}) does not admit any LHVM in Eq.~(\ref{eq:BiLHVM}), the correlation $P(ab|xy)$ is a nonlocal quantum correlation and the state $\rho$ is said to be a nonlocal quantum state. Otherwise, the correlation $P(ab|xy)$ is a local quantum correlation and $\rho$ is said to be a local quantum state.

The violation of the CHSH inequality is a sufficient evidence that the state $\rho$ is nonlocal. A local quantum state $\rho$ must satisfy any CHSH inequality as follows:
\begin{align}
   {\rm tr}(\rho S) \le 2,
\end{align}
where $S$ is a CHSH operator of the form
\begin{align}
   S = A \otimes (B + B') + A' \otimes (B - B').
\end{align}
Here $A$ and $A'$ are measurement observables of Alice's qubit while $B$ and $B'$ are that of Bob's qubit. If there is any CHSH operator such that the inequality is violated, the state must be nonlocal.

The nonlocality of the two-qubit state $\rho$ can be quantified by the maximal value of the CHSH expression as follows:
\begin{align}
   \mathcal{N}(\rho) = \max_{S} {\rm tr}(\rho S),
\end{align}
where the maximum is over all possible CHSH operators. In the following discussion, we refer to the nonlocality of a state as the quantity $\mathcal{N}(\cdot)$. We also refer to the CHSH operator that can achieve $\mathcal{N}(\rho)$ as the optimal CHSH operator of the state $\rho$.

It is obvious that the measure $\mathcal{N}(\cdot)$ of bipartite nonlocality is a convex function of bipartite density operators on $\mathcal{H}_{2}^{\otimes 2}$. Suppose $\rho^{AB}$ is a density operator that can be generated from the ensemble $\{q_{k}, \rho_{k}^{AB}\}$, namely, $\rho^{AB}=\sum_{k} q_{k} \rho_{k}^{AB}$. Let $S$ be the optimal CHSH operator of the state $\rho^{AB}$. Then, we are convinced of the convexity of nonlocality by the following relation:
\begin{align}\label{ineq:convexityNL}
   \mathcal{N}\big(\sum_{k} q_{k} \rho_{k}^{AB}\big) = \sum_{k} q_{k} {\rm tr}(\rho_{k}^{AB}S) \le \sum_{k}q_{k} \mathcal{N}(\rho_{k}^{AB}).
\end{align}
The equality in Eq.~(\ref{ineq:convexityNL}) holds if and only if the operator $S$ also acts as the optimal CHSH operator for all constituent states $\rho_{k}^{AB}$.

\subsection{The optimal CHSH operator}\label{sec:optimalCHSHoperator}
Horodecki \textit{et al}. got an analytical expression for the nonlocality of a general two-qubit state~\cite{RH95}. In this section, we review the process of analyzing the nonlocality and further work out the mathematical expression of the corresponding optimal CHSH operator.

Any density operator $\rho$ of a two-qubit quantum system can be represented by the combination of the identity operator and the generators of the SU(2) algebra~\cite{FJ81} as follows:
\begin{align}
   \rho = \frac{1}{4} (I \otimes I + \vec{r}\cdot \vec{\sigma} \otimes I + I \otimes \vec{s}\cdot \vec{\sigma} + \sum_{j,k=1}^{3} T_{jk} \sigma_{j} \otimes \sigma_{k}),
\end{align}
where $\vec{\sigma} = (\sigma_{1}, \sigma_{2}, \sigma_{3})$ is the vector of Pauli matrices and the coefficient $T_{jk} = {\rm tr}(\sigma_{j}\otimes \sigma_{k} \rho)$ is the element of a $3 \times 3$ real matrix $T$. The vectors $\vec{r}$ and $\vec{s}$ only determine the properties of the two subsystems, respectively. The global properties between the joint system are contained in the matrix $T$. Thus, we call $T$ the correlation matrix of the bipartite state $\rho$.

Note that any measurement observable $A$ for a qubit system can be represented by a unit column vector $\vec{a} \in \mathbb{R}^{3}$ such that $A = \vec{a}\cdot \vec{\sigma}$. Let $\vec{a}$, $\vec{a'}$, $\vec{b}$, and $\vec{b'}$ be the corresponding vectors of the measurement observables $A$, $A'$, $B$, and $B'$, respectively. Then, the expectation of the joint measurement $A\otimes B$ can be written as
\begin{align}
   {\rm tr}(A\otimes B \rho) = \sum_{j,k=1}^{3} a_{j}T_{jk}b_{k} = \langle a|T|b\rangle.
\end{align}
Here we abuse the Dirac notation $|a\rangle$, which usually stands for the pure state of a quantum system, for the unit column vector $\vec{a}$. In the following discussions, we interchange the two kinds of notations based on the convenience of expression.

Let $\vec{b} + \vec{b'} = 2 \cos{\theta}\vec{c}$ and $\vec{b} - \vec{b'} = 2 \sin{\theta}\vec{c'}$  for $\theta\in [0,\pi)$ and $\vec{c},\vec{c'}\in \mathbb{R}^{3}$. Note that $\vec{c}$ and $\vec{c'}$ are orthogonal unit vectors. Let $T^{T}T = \sum_{k=1}^{3} \lambda_{k} |\mu_{k}\rangle \langle \mu_{k}|$ be the spectral decomposition of the matrix $T^{T}T$. Without loss of generality, suppose the eigenvalues  $\lambda_{k}$ are in nonincreasing order. Then, we can get
\begin{align}
   {\rm tr}(S\rho) & = {\rm tr}(A\otimes(B+B')\rho) + {\rm tr}(A'\otimes(B - B')\rho) \nonumber \\
             & = 2 \cos{\theta} \langle a|T|c\rangle + 2 \sin{\theta} \langle a'|T|c'\rangle \nonumber \\
             & \le 2 \sqrt{(\langle a|T|c\rangle)^{2} + ( \langle a'|T|c'\rangle)^{2}} \label{ineq:first} \\
             & \le 2 \sqrt{\| T|c\rangle\|^{2} + \| T|c'\rangle\|^{2}} \label{ineq:second} \\
             & \le 2\sqrt{\lambda_{1} + \lambda_{2}}. \label{ineq:third}
\end{align}
The inequality in Eq.~(\ref{ineq:third}) is with equation when $|c\rangle = \eta|\mu_{1}\rangle$ and $|c'\rangle = \eta'|\mu_{2}\rangle$ with $\eta,\eta'\in\{+1, -1\}$ being signs. In this case, $\|T|c\rangle\|^{2} = \langle c|T^{T}T|c\rangle = \lambda_{1}$ and $\|T|c'\rangle\|^{2} = \langle c'|T^{T}T|c'\rangle = \lambda_{2}$ . The equality of Eq.~(\ref{ineq:second}) holds when $|a\rangle = \delta \frac{T|c\rangle}{\|T|c\rangle\|}$ and $|a'\rangle = \delta \frac{T|c'\rangle}{\|T|c'\rangle\|}$ with signs $\delta, \delta' \in \{+1, -1\}$. And the equality in Eq.~(\ref{ineq:first}) holds when $\cos{\theta} = \frac{\langle a|T|c\rangle}{\sqrt{\lambda_{1} + \lambda_{2}}}$ and $\sin{\theta} = \frac{\langle a'|T|c'\rangle}{\sqrt{\lambda_{1} + \lambda_{2}}}$. Thus, we get that the nonlocality of $\rho$ is $\mathcal{N}(\rho) = 2\sqrt{\lambda_{1} + \lambda_{2}}$. The nonlocality can be achieved when all three equalities in Eqs.~(\ref{ineq:first})--(\ref{ineq:third}) hold true simultaneously. The corresponding optimal CHSH operator of state $\rho$ is
\begin{align}
   S & = 2 \cos{\theta} A\otimes C + 2\sin{\theta} A' \otimes C' \nonumber\\
     & = 2 \delta\sqrt{\frac{\lambda_{1}}{\lambda_{1} + \lambda_{2}}} (\frac{\delta\eta}{\sqrt{\lambda_{1}}}T\vec{\mu_{1}}) \cdot \vec{\sigma} \otimes (\eta\vec{\mu_{1}}) \cdot \vec{\sigma} \nonumber\\
     & \quad + 2 \delta' \sqrt{\frac{\lambda_{2}}{\lambda_{1} + \lambda_{2}}} (\frac{\delta'\eta'}{\sqrt{\lambda_{2}}}T\vec{\mu_{2}}) \cdot \vec{\sigma} \otimes (\eta'\vec{\mu_{2}}) \cdot \vec{\sigma} \nonumber\\
     & = \frac{2}{\sqrt{\lambda_{1} + \lambda_{2}}} (T\vec{\mu_{1}} \cdot \vec{\sigma} \otimes \vec{\mu_{1}}\cdot \vec{\sigma} + T\vec{\mu_{2}} \cdot \vec{\sigma} \otimes \vec{\mu_{2}}\cdot \vec{\sigma}). \label{eq:optimalCHSHoperator}
\end{align}
Let $w_{jk} \equiv \frac{1}{4}{\rm tr}(S\sigma_{j} \otimes \sigma_{k})$ be the Pauli coefficients of the optimal CHSH operator $S$. Namely, $S = \sum_{j,k = 1}^{3} w_{jk} \sigma_{j} \otimes \sigma_{k}$. Thus, the optimal CHSH operator $S$ of state $\rho$ is uniquely corresponding to a $3\times 3$ real matrix $W = (w_{jk})_{3\times 3}$. Taking the standard basis, the matrix $W$ can be written as
\begin{align}
   W & = \sum_{j,k = 1}^{3} w_{jk} |j\rangle \langle k| \nonumber \\
     & = \frac{2}{\sqrt{\lambda_{1} + \lambda_{2}}} \sum_{j,k= 1}^{3} ((T|\mu_{1}\rangle)_{j} (|\mu_{1}\rangle)_{k} + (T|\mu_{2}\rangle)_{j} (|\mu_{2}\rangle)_{k}) |j\rangle \langle k| \nonumber \\
     & = \frac{2}{\sqrt{\lambda_{1} + \lambda_{2}}} T(|\mu_{1}\rangle \langle \mu_{1}| + |\mu_{2}\rangle \langle \mu_{2}|) . \label{eq:coeffMatrix4OptCHSHoperator}
\end{align}
The above analysis is based on Horodecki \textit{et al.}'s analysis for obtaining the maximum CHSH violation of a general two-qubit state~\cite{RH95}. Our contribution is to conclude the optimal CHSH operator in Eq.~(\ref{eq:optimalCHSHoperator}) to achieve Horodecki \textit{et al.}'s maximum CHSH violation and reformulate it to be expressed by a $3\times 3$ coefficient matrix under the Pauli matrices basis in Eq.~(\ref{eq:coeffMatrix4OptCHSHoperator}). As the extended results are to be used in the following discussion, we conclude them as the following lemma.
\begin{lemma}\label{lemma:OptimalCHSHOperator}
   Suppose $\rho$ is an arbitrary two-qubit state with correlation matrix $T$ and the eigenvalues in the spectral decomposition $T^{T}T = \sum_{k=1}^{3} \lambda_{k} |\mu_{k}\rangle \langle \mu_{k}|$ are in nonincreasing order. Then, the nonlocality of state $\rho$ is $\mathcal{N}(\rho) = 2\sqrt{\lambda_{1} + \lambda_{2}}$. The optimal CHSH operator of state $\rho$ is $S= \sum_{j,k = 1}^{3} w_{jk} \sigma_{j} \otimes \sigma_{k}$, which is uniquely specified by the corresponding Pauli coefficients matrix
   \begin{align}
      W  = \frac{2}{\sqrt{\lambda_{1} + \lambda_{2}}} T(|\mu_{1}\rangle \langle \mu_{1}| + |\mu_{2}\rangle \langle \mu_{2}|) .
\end{align}
\end{lemma}

\section{Two different states having the same optimal CHSH operator} \label{sec:conditionSameOptimalCHSHOperator}
In this section, we research the question whether the nonlocality of different quantum states can be optimally stimulated by the same quantum nonlocality test setting. Let $T$ and $F$ be correlation matrices of two-qubit states $\rho$ and $\varrho$, respectively.
We try to work out the condition that the optimal CHSH operators of the states $\rho$ and $\varrho$ can be the same.

Suppose the singular value decomposition of $T$ is $T=UDV$, where $U = \sum_{k=1}^{3}|\mu_{k}\rangle \langle k|$ and $V = \sum_{k=1}^{3}|k\rangle \langle \nu_{k}|$ are orthogonal matrices and $D = \sum_{k=1}^{3}t_{k}|k\rangle \langle k|$ is a diagonal matrix. Then, we have $T=\sum_{k=1}^{3}t_{k}|\mu_{k}\rangle \langle v_{k}|$ and $T^{T}T=\sum_{k=1}^{3}t_{k}^{2}|v_{k}\rangle \langle v_{k}|$. With out loss of generality, suppose $|t_{1}|\ge |t_{2}|\ge |t_{3}|$. According to Lemma~\ref{lemma:OptimalCHSHOperator}, the corresponding Pauli coefficients matrix of the optimal CHSH operator for state $\rho$ can be written as
\begin{align}
   W = \frac{2}{\sqrt{t_{1}^{2} + t_{2}^{2}}} (t_{1} |\mu_{1}\rangle \langle \nu_{1}| + t_{2} |\mu_{2}\rangle \langle \nu_{2}|).
\end{align}
Suppose $F=PEQ$, where $P = \sum_{k=1}^{3} |\alpha_{k}\rangle \langle k|$, $Q = \sum_{k=1}^{3} |k\rangle \langle \beta_{k}|$, and $E = \sum_{k=1}^{3} f_{k} |k\rangle \langle k|$. Further suppose $|f_{k}|$ are in nonincreasing order. Similarly, the corresponding Pauli coefficients matrix of the optimal CHSH operator  for state $\varrho$ can be written as
\begin{align}
   R = \frac{2}{\sqrt{f_{1}^{2} + f_{2}^{2}}} (f_{1} |\alpha_{1}\rangle \langle \beta_{1}| + f_{2} |\alpha_{2}\rangle \langle \beta_{2}|).
\end{align}

The preset condition that the optimal CHSH operators of two-qubit states $\rho$ and $\varrho$ are the same is equivalent to the relation $W=R$. Let $\cos{t} = \frac{t_{1}}{\sqrt{t_{1}^{2} + t_{2}^{2}}}$, $\sin{t} = \frac{t_{2}}{\sqrt{t_{1}^{2} + t_{2}^{2}}}$, $\cos{f} = \frac{f_{1}}{\sqrt{f_{1}^{2} + f_{2}^{2}}}$, and $\sin{f} = \frac{f_{2}}{\sqrt{f_{1}^{2} + f_{2}^{2}}}$. Then, the relation $W=R$ can be equivalently written as
\begin{align}
     & \cos{t} |\mu_{1}\rangle \langle \nu_{1}| + \sin{t}|\mu_{2}\rangle \langle \nu_{2}| \nonumber \\
   = & \cos{f} |\alpha_{1}\rangle \langle \beta_{1}| + \sin{f} |\alpha_{2}\rangle \langle \beta_{2}|.
\end{align}
Let $M=(m_{kj})$ be a $2\times 2$ matrix where $m_{kj} = \langle \mu_{k}|\alpha_{j} \rangle \langle \beta_{j}|\nu_{k}\rangle$. Then, we get
\begin{align}
   \cos{t} & = m_{11} \cos{f} + m_{12} \sin{f}, \label{eq:cost}\\
   \sin{t} & = m_{21} \cos{f} + m_{22} \sin{f}, \\
   \cos{f} & = m_{11} \cos{t} + m_{21} \sin{t}, \\
   \sin{f} & = m_{12} \cos{t} + m_{22} \sin{t}. \label{eq:sinf}
\end{align}
Denote $|t\rangle \equiv \cos{t}|0\rangle + \sin{t}|1\rangle$ and $|f\rangle \equiv \cos{f}|0\rangle + \sin{f}|1\rangle$. Equations.~(\ref{eq:cost})--(\ref{eq:sinf}) can be equivalently written as
\begin{align}
  |t\rangle = M|f\rangle \text{ and } |f\rangle = M^{T}|t\rangle.
\end{align}
It follows that $|t\rangle = MM^{T}|t\rangle$ and $|f\rangle = M^{T}M|f\rangle$ for all possible vectors $|t\rangle$ and $|f\rangle$. Thus, it must have $M^{T}M = MM^{T} = I$, which means $M$ is an orthogonal matrix. Namely, $ m_{11}^{2} + m_{12}^{2} = 1$ and $ m_{21}^{2} + m_{22}^{2} = 1$.

However, we note that
\begin{align}
     m_{11}^{2} + m_{12}^{2}  = (\langle \mu_{1}|\alpha_{1} \rangle \langle \beta_{1}|\nu_{1}\rangle)^{2} +(\langle \mu_{1}|\alpha_{2} \rangle \langle \beta_{2}|\nu_{1}\rangle)^{2} \le 1 
       \label{ineq:m11m12}
\end{align}
and
\begin{align}
     m_{21}^{2} + m_{22}^{2}  = (\langle \mu_{2}|\alpha_{1} \rangle \langle \beta_{1}|\nu_{2}\rangle)^{2} +(\langle \mu_{2}|\alpha_{2} \rangle \langle \beta_{2}|\nu_{2}\rangle)^{2} \le 1.  \label{ineq:m21m22} 
\end{align}
Thus, the equalities in Eqs.~(\ref{ineq:m11m12}) and (\ref{ineq:m21m22}) must hold simultaneously. The conditions for which the equalities hold true are $|\mu_{k}\rangle = \delta_{k} |\alpha_{k}\rangle$ and $|\nu_{k}\rangle = \delta_{k}' |\beta_{k}\rangle$ for $k=1,2$ or $|\mu_{1}\rangle = \delta_{1} |\alpha_{2}\rangle$, $|\mu_{2}\rangle = \delta_{2} |\alpha_{1}\rangle$, $|\nu_{1}\rangle = \delta_{1}' |\beta_{2}\rangle$, and $|\mu_{2}\rangle = \delta_{2} |\alpha_{1}\rangle$ where $\delta_{k},\delta_{k}' \in \{\pm 1\}$ are signs.
With out loss of generality, we can set $\delta_{k}=\delta_{k}' = 1$ and pat the signs into the singular values. In the former case of conditions, it follows that $|\mu_{k}\rangle = |\alpha_{k}\rangle$ and $|\nu_{k}\rangle =  |\beta_{k}\rangle$ for $k=1,2,3$ and $\cos{t} = \cos{f}$, $\sin{t} = \sin{f}$.  Hence, the correlation matrices of states $\rho$ and $\varrho$ are related as follows:
\begin{align}
   T & = \sum_{k=1}^{3} t_{k} |\mu_{k}\rangle \langle \nu_{k} | \nonumber \\
     & = \sqrt{t_{1}^{2} + t_{2}^{2}}( \cos{t} |\alpha_{1}\rangle \langle \beta_{1} | + \sin{t} |\alpha_{2}\rangle \langle \beta_{2} |) + t_{3} |\alpha_{3}\rangle \langle \beta_{3} |  \nonumber \\
     & = \sqrt{\frac{t_{1}^{2} + t_{2}^{2}}{f_{1}^{2} + f_{2}^{2}}}( f_{1} |\alpha_{1}\rangle \langle \beta_{1} | + f_{2} |\alpha_{2}\rangle \langle \beta_{2} |) + t_{3} |\alpha_{3}\rangle \langle \beta_{3} |. \label{eq:2statesameScond}
\end{align}
If the second condition holds, we get the same relation as the one shown in Eq.~(\ref{eq:2statesameScond}).

A more general case is that $|t_{1}|$, $|t_{2}|$, and $|t_{3}|$ are in a specified order, say, decreasing order, while $|f_{1}|$, $|f_{2}|$, and $|f_{3}|$ are in an arbitrary order. Without loss of generality, suppose $|f_{2}|$ is the minimum in $\{|f_{k}|\}$. With similar analysis, we find out that the relation $W=R$ can be equivalently written as
\begin{align}
  T = \sqrt{\frac{t_{1}^{2} + t_{2}^{2}}{f_{1}^{2} + f_{2}^{2}}}( f_{1} |\alpha_{1}\rangle \langle \beta_{1} | + f_{3} |\alpha_{3}\rangle \langle \beta_{3} |) + t_{3} |\alpha_{2}\rangle \langle \beta_{2} |.
\end{align}
Obviously, the absolute values of the correlation matrices $T$ and $F$ singular values are in the same order.

Based on the above analysis, we can conclude the conditions for the situation that two different general two-qubit states have the same optimal CHSH operator in the following theorem.

\begin{theorem}\label{thm:conditionforsameCHSHoperator}
   Suppose $T$ is the correlation matrix of the two-qubit state $\rho$ and $F$ is that of the two-qubit state $\varrho$. Let $t_1$, $t_2$, $t_3$ and $f_1$, $f_2$, $f_3$ be singular values of the correlation matrices $T$ and $F$, respectively. Then, the states $\rho$ and $\varrho$ have the same optimal CHSH operator if and only if they satisfy three conditions, which are listed as follows:
   \begin{enumerate}[(1)]
     \item $T$ and $F$ have the same orthogonal matrices in the singular value decompositions. Namely, the singular value decomposition of $T$ and $F$ can be written as
         \begin{align}
            T = UDV \text{  and   } F = UEV,
         \end{align}
         where $D = \sum_{k=1}^{3} t_{k}|k\rangle \langle k|$, $E = \sum_{k=1}^{3} f_{k}|k\rangle \langle k|$, and $U$ and $V$ are orthogonal matrices.

     \item The absolute values of two sets of singular values, $\{|t_{k}|\}$ and $\{|f_{k}|\}$, are of the same order.

     \item Without loss of generality, suppose $\{|t_{k}|\}$ and $\{|f_{k}|\}$ are in nonincreasing order. Then, it should have
         \begin{align}
            \frac{t_{1}}{f_{1}} = \frac{t_{2}}{f_{2}} = \sqrt{\frac{t_{1}^{2} + t_{2}^{2}}{f_{1}^{2} + f_{2}^{2}}}.
         \end{align}

   \end{enumerate}
\end{theorem}

\section{The collection of two-qubit pure states that share the same optimal CHSH operator}\label{sec:2qubitSameOptimalCHSHOperator}
Recall that any mixed state can be viewed as an ensemble of pure states. Suppose $\rho$ is an arbitrary two-qubit state which is generated from the ensemble $\{ p_{k}, |\psi_{k}\rangle \}$. Namely, $\rho = \sum_{k}p_{k} |\psi_{k}\rangle \langle \psi_{k}|$, where $p_{k}\ge 0$ and $\sum_{k}p_{k} = 1$. The convexity of nonlocality indicates that
\begin{align}\label{eq:convexityofpure2QB}
    \mathcal{N}(\rho)  \le \sum_{k}p_{k} \mathcal{N}(\psi_{k}).
\end{align}
The equality holds if and only if the optimal CHSH operator of $\rho$ also acts as the optimal CHSH operator of $|\psi_{k}\rangle$ for all $k$. In other words, it requires that the collection of pure states $\{|\psi_{k}\rangle \}$ has the same optimal CHSH operator.

In this section, we exploit the collection of two-qubit pure states which have the same optimal CHSH operator.

\subsection{The role of local unitary operations}
Suppose a general two-qubit state $\rho$ is shared by Alice and Bob. $U_{A}$ and $U_{B}$ are unitary operators which act on Alice's and Bob's subsystems, respectively. They can transform the state $\rho$ into another state
\begin{align}
   \rho' = (U_{A} \otimes U_{B}) \rho (U_{A} \otimes U_{B})^{\dagger}.
\end{align}
Suppose the optimal CHSH operators for $\rho$ and $\rho'$ are $S$ and $S'$, respectively. Then, it follows that
\begin{align}
   \mathcal{N}(\rho') & = {\rm tr}(S'\rho') \nonumber \\
       & = {\rm tr}((U_{A} \otimes U_{B})^{\dagger}S'(U_{A} \otimes U_{B}) \rho) \nonumber \\
       & \le \mathcal{N}(\rho) \nonumber \\
       & = {\rm tr}(S\rho). \nonumber
\end{align}
Similarly, it has ${\rm tr}(S\rho) \le {\rm tr}((U_{A} \otimes U_{B})^{\dagger}S'(U_{A} \otimes U_{B}) \rho)$. Thus, the equality should hold, namely, ${\rm tr}(S\rho) = {\rm tr}((U_{A} \otimes U_{B})^{\dagger}S'(U_{A} \otimes U_{B}) \rho)$. Therefore, it should have $\mathcal{N}(\rho') = \mathcal{N}(\rho)$ and $S' = (U_{A} \otimes U_{B})S(U_{A} \otimes U_{B})^{\dagger}$. It also indicates that local unitary operations do not affect the nonlocality of the shared state.

Suppose the correlation matrices of $\rho$ and $\rho'$ are $T=(T_{kj})_{3\times 3}$ and $T'=(T'_{kj})_{3\times 3}$, respectively. Then, any element $T_{kj}'$ of the correlation matrix $T'$ can be equivalently rewritten as
\begin{align}
   T_{kj}' = &  {\rm tr}(\rho' \sigma_{k} \otimes \sigma_{j}) \nonumber \\
           = & \frac{1}{4} \sum_{k',j' = 1}^{3} T_{k'j'}{\rm tr}(\sigma_{k'}U_{A}^{\dagger} \sigma_{k}U_{A}) {\rm tr}(\sigma_{j'}U_{B}^{\dagger} \sigma_{j}U_{B}). \nonumber
\end{align}
Let $R_{A}$ and $R_{B}$ be two matrices with elements defined as $(R_A)_{kk'} \equiv \frac{1}{2}{\rm tr}(\sigma_{k}U_{A}\sigma_{k'}U_{A}^{\dagger})$ and $(R_B)_{jj'} \equiv \frac{1}{2}{\rm tr}(\sigma_{j}U_{B}\sigma_{j'}U_{B}^{\dagger})$, respectively.
It follows that $T_{kj}' = \sum_{k',j' = 1}^{3} (R_A)_{kk'} T_{k'j'}(R_B)_{jj'} = (R_A TR_B^T)_{kj}$. Thus, it has
\begin{align}
   T' = R_A TR_B^T.
\end{align}
Note that SU(N) matrices have the completeness relation as follows~\cite{JS95}:
\begin{align}\label{eq:SUNcompleteness}
   \sum_{j=1}^{N^{2} - 1} (\sigma_{j})_{ki}(\sigma_{j})_{mn} = 2\Delta_{im}\Delta_{kn} - \frac{2}{N}\delta_{ki}\Delta_{mn},
\end{align}
where $\Delta_{jk} = 1$ if $j=k$ and  $\Delta_{jk} = 0$ otherwise.
Applying the completeness relation in Eq.~(\ref{eq:SUNcompleteness}), it can be shown that $R^{A}$ and $R^{B}$ are orthogonal matrices.

We can conclude the above analysis as the following lemma.
\begin{lemma}\label{lm:LUonC.M.}
   Suppose Alice and Bob apply the local unitaries $U_A$ and $U_B$ on subsystems of the shared two-qubit state $\rho$, respectively. Suppose the correlation matrix of the state $\rho$ is $T$. Then, the correlation matrix of the final state can be written as
   \begin{align}
      T' = R_A TR_B^T,
   \end{align}
   where $R_{A}$ and ${R_{B}}^{T}$ are two $3\times 3$  orthogonal real matrices defined as
   \begin{align}\label{eq:C.M.RotationA}
      (R_A)_{kk'} \equiv &  \frac{1}{2}{\rm tr}(\sigma_{k}U_{A}\sigma_{k'}U_{A}^{\dagger})
   \end{align}
   and
   \begin{align}\label{eq:C.M.RotationB}
      (R_B)_{jj'} \equiv &  \frac{1}{2}{\rm tr}(\sigma_{j}U_{B}\sigma_{j'}U_{B}^{\dagger}),
   \end{align}
   respectively.
\end{lemma}

\subsection{The class of two-qubit pure state pairs with the same optimal CHSH operators}
We consider a general two-qubit pure state $|\psi\rangle$ with correlation matrix $T$. Suppose the singular value decomposition of $T$ is $T=R_{A}DR_{B}$, where $D$ is a diagonal matrix and $R_{A}$ and $R_{B}$ are orthogonal matrices.
Let $U_{A}$ and $U_{B}$ be the unitary operators that generate the orthogonal matrices $R_{A}^{T}$ and $R_{B}$ according to the definitions in Eqs.~(\ref{eq:C.M.RotationA}) and (\ref{eq:C.M.RotationB}), respectively. According to Lemma~\ref{lm:LUonC.M.}, the correlation matrix of the state $(U_A\otimes U_B)|\psi\rangle$ is the diagonal matrix $D$. Namely, any two-qubit pure state can be transformed into another pure state with diagonal correlation matrix via local unitary operations.

Suppose $|\varphi\rangle$ is a two-qubit pure state which has the same optimal CHSH operator with $|\psi\rangle$. Combining Theorem~\ref{thm:conditionforsameCHSHoperator}, it is obvious that any two-qubit pure states $|\psi\rangle$ and $|\varphi\rangle$ can be simultaneously transformed into the states which have diagonal correlation matrices, by applying the same local unitary operations. Thus, we only need to look into the kind of two-qubit states with diagonal correlation matrices in order to analyze the class of two-qubit pure state pairs which have the same optimal CHSH operators.

Up to an ignorable global phase, a general two-qubit pure state can be equivalently written as
\begin{align}
   |\psi\rangle = a|00\rangle + (b_{1} + ib_{2})|01\rangle + (c_{1} + ic_{2})|01\rangle + (d_{1} + id_{2})|01\rangle  \nonumber
\end{align}
where the parameters are all real numbers that satisfy the unit condition $a^2 + b_1^2 + b_2^2 + c_1^2 + c_2^2 + d_1^2 + d_2^2 = 1$. The elements of the corresponding correlation matrix $T=(T_{ij})_{3\times 3}$ are
\begin{align}
   T_{12} & = 2(-ad_{2} - b_{2}c_{1} + b_{1}c_{2}), \nonumber \\
   T_{13} & = 2(ac_{1} - b_{1}d_{1} - b_{2}d_{2}), \nonumber \\
   T_{21} & = 2(-ad_{2} + b_{2}c_{1} - b_{1}c_{2}), \nonumber \\
   T_{23} & = 2(-ac_{2} - b_{2}d_{1} + b_{1}d_{2}), \nonumber \\
   T_{31} & = 2(ab_{1} - c_{1}d_{1} - c_{2}d_{2}),  \nonumber \\
   T_{32} & = 2(-ab_{2} - c_{2}d_{1} + c_{1}d_{2}), \nonumber \\
   T_{11} & = 2(ad_{1} + b_{1}c_{1} + b_{2}c_{2}), \nonumber \\
   T_{22} & = 2(-ad_{1} + b_{1}c_{1} + b_{2}c_{2}),   \nonumber \\
   T_{33} & = 2(a^{2} + d_{1}^{2} + d_{2}^{2}) - 1. \nonumber
\end{align}
By setting $T_{kj} = 0$ for all $k\ne j \in \{1,2,3\}$, we derive that any two-qubit pure state, of which the correlation matrix is diagonal, should be in one of the forms listed as follows:
\begin{align}
   |\gamma_{(\theta)}\rangle & = \cos{\theta} |00\rangle + \sin{\theta} |11\rangle, \label{eq:stateDiagonalCM1}  \\
   |\omega_{(\theta)}\rangle & = \cos{\theta} |01\rangle + \sin{\theta} |10\rangle,  \label{eq:stateDiagonalCM2} \\
   |\lambda_{(\theta,\delta)}\rangle & = \cos{\theta} (|00\rangle + \delta |11\rangle)/\sqrt{2} + i\sin{\theta}(|01\rangle - \delta|10\rangle)/\sqrt{2}, \label{eq:stateDiagonalCM3} \\
   |\phi_{(\theta,\delta)} \rangle & = \cos{\theta} (|00\rangle + \delta |11\rangle)/\sqrt{2} + \sin{\theta}(|01\rangle + \delta|10\rangle)/\sqrt{2}, \label{eq:stateDiagonalCM4}
\end{align}
where $\theta\in[0,\pi)$ and $\delta = \pm 1$. The corresponding correlation matrices are
\begin{align}
   T_{\gamma}(\theta) & =\left[\begin{array}{ccc}
      \sin{2\theta} & 0 & 0 \\
      0 &  -\sin{2\theta} & 0 \\
      0 & 0 & 1
   \end{array}\right], \nonumber \\
   T_{\omega}(\theta) & =\left[\begin{array}{ccc}
      \sin{2\theta} & 0 & 0 \\
      0 &  \sin{2\theta} & 0 \\
      0 & 0 & -1
   \end{array}\right], \nonumber \\
   T_{\alpha}(\theta,\delta) & =\left[\begin{array}{ccc}
      \delta\cos{2\theta} & 0 & 0 \\
      0 &  -\delta & 0 \\
      0 & 0 & \cos{2\theta}
   \end{array}\right], \nonumber \\
   T_{\beta}(\theta,\delta) & =\left[\begin{array}{ccc}
      \delta & 0 & 0 \\
      0 &  -\delta\cos{2\theta} & 0 \\
      0 & 0 & \cos{2\theta}
   \end{array}\right]. \nonumber
\end{align}
Applying Theorem~\ref{thm:conditionforsameCHSHoperator}, we find out that there are no pairs of two-qubit pure states which are of different forms than in Eqs.~(\ref{eq:stateDiagonalCM1})--(\ref{eq:stateDiagonalCM4}) and have the same optimal CHSH operator.
However, there are pairs of states which are of the same form as in Eqs.~(\ref{eq:stateDiagonalCM1})--(\ref{eq:stateDiagonalCM4}) and have the same optimal CHSH operator. The collection of such state pairs is listed as follows:
\begin{align}
   \Gamma  \equiv & \{(|\gamma_{(\theta)}\rangle, |\gamma_{(\frac{\pi}{2} - \theta)}\rangle) \},   \nonumber \\
   \Omega  \equiv & \{(|\omega_{(\theta)}\rangle, |\omega_{(\frac{\pi}{2} - \theta)}\rangle) \},  \nonumber  \\
   \Lambda \equiv & \{(|\lambda_{(\theta,\delta)}\rangle, |\lambda_{(-\theta,\delta)}\rangle) \}, \nonumber  \\
   \Phi    \equiv & \{(|\phi_{(\theta,\delta)}\rangle, |\phi_{(-\theta,\delta)}\rangle) \}.  \nonumber
\end{align}
Moreover, we could not find any three two-qubit pure states which are of the same form and have the same optimal CHSH operator. Thus, at most two different two-qubit pure states could have the same optimal CHSH operator.

As any two-qubit pure state can be written as $|\gamma_{(\theta)}\rangle = \cos{\theta} |00\rangle + \sin{\theta} |11\rangle$ up to some local unitaries, we can get any state of the types $|\omega_{(\theta)}\rangle$, $|\lambda_{(\theta,\delta)}\rangle$, or $|\phi_{(\theta,\delta)} \rangle$ by applying the corresponding local unitary operations on some state of the type $|\gamma_{(\theta)}\rangle$. Thus, we only need to consider the collection $\Gamma$. By applying Lemma~\ref{lemma:OptimalCHSHOperator}, we find out that the two-qubit pure states $|\gamma_{(\theta)}\rangle$ and $|\gamma_{(\frac{\pi}{2}-\theta)}\rangle$ share the same optimal CHSH operators $S_{\theta +} = \frac{2}{\sqrt{1+\sin^{2}{2\theta}}}(\sin{2\theta} \sigma_{x} \otimes \sigma_{x} + \sigma_{z} \otimes \sigma_{z})$ and $S_{\theta -} = \frac{2}{\sqrt{1+\sin^{2}{2\theta}}}(-\sin{2\theta} \sigma_{y} \otimes \sigma_{y} + \sigma_{z} \otimes \sigma_{z})$.

We conclude the above analysis as the following theorem, which is the main result of this section.
\begin{theorem}\label{thm:pure2QBsameOCHSHOperator}
   Two different two-qubit pure states $|\psi\rangle$ and $|\psi'\rangle$ have the same optimal CHSH operator if and only if they can be written as
   \begin{align}
      |\psi\rangle & = (U_{A}\otimes U_{B})(\cos{\theta} |00\rangle + \sin{\theta} |11\rangle) \text{ and} \nonumber\\
      |\psi'\rangle & = (U_{A}\otimes U_{B})(\sin{\theta} |00\rangle + \cos{\theta} |11\rangle), \nonumber
   \end{align}
   where $U_{A}, U_{B}$ are some unitaries on $\mathcal{H}_{2}$ and $\theta\in[0,\pi)$. The operators $(U_{A}\otimes U_{B})S_{\theta +}(U_{A}\otimes U_{B})^{\dagger}$ and $(U_{A}\otimes U_{B})S_{\theta -}(U_{A}\otimes U_{B})^{\dagger}$ act as the the optimal CHSH operators for both $|\psi\rangle$ and $|\psi'\rangle$, where
   \begin{align}
      S_{\theta +} & = \frac{2}{\sqrt{1+\sin^{2}{2\theta}}}(\sin{2\theta} \sigma_{x} \otimes \sigma_{x} + \sigma_{z} \otimes \sigma_{z}) \text{ and} \nonumber \\
      S_{\theta -} & = \frac{2}{\sqrt{1+\sin^{2}{2\theta}}}(-\sin{2\theta} \sigma_{y} \otimes \sigma_{y} + \sigma_{z} \otimes \sigma_{z}). \nonumber
   \end{align}
\end{theorem}

\section{Nonlocality upper bounded by entanglement}\label{sec:NLupperbound}
In this section, we show that the nonlocality of a general two-qubit state is upper bounded by a function of the corresponding entanglement. We further figure out the class of two-qubit states of which the upper bound of nonlocality can be reached.

Recall that any pure state of a two-qubit system can be written as
\begin{align}
   |\Phi_{\theta}\rangle = \cos{\theta} |00\rangle + \sin{\theta}|11\rangle, \theta\in[0,\frac{\pi}{4}]
\end{align}
upper to some local unitary transformations. As the entanglement and nonlocality are invariant under local unitary transformations, it is sufficient to investigate the state $|\Phi_{\theta}\rangle$ in order to qualitatively analyze the entanglement and nonlocality of an arbitrary two-qubit pure state.

In the previous section, we have figured out that the correlation matrix of the state $|\Phi_{\theta}\rangle$ is a diagonal matrix with diagonal entries being $\sin{2\theta}$, $-\sin{2\theta}$, and 1. By applying Lemma~\ref{lemma:OptimalCHSHOperator}, it is obvious that the nonlocality of the state $|\Phi_{\theta}\rangle$ is
\begin{align}
   \mathcal{N}(\Phi_{\theta}) = 2\sqrt{1 + \sin{2\theta}^{2}}.
\end{align}

According to the definition in Eq.~(\ref{eq:concurrence2qubitpure}), we get the concurrence of the state $|\Phi_{\theta}\rangle$  as follows
\begin{align}
   C(\Phi_{\theta}) = |{\rm tr}(|\Phi_{\theta}\rangle \langle \Phi_{\theta}| Y \otimes Y)| = |T_{22}| = |\sin{2\theta}|.
\end{align}

Thus, it is straightforward to get the relationship between the nonlocality and entanglement of a two-qubit pure state. We conclude the above analysis as the following lemma.
\begin{lemma}\label{lemma:NLandCC2QBpure}
   Suppose $|\psi\rangle$ is an arbitrary pure state of a two-qubit system. The nonlocality and concurrence of $|\psi\rangle$ are related by the following equation:
   \begin{align}
      \mathcal{N}(\psi) = 2\sqrt{1 + C(\psi)^{2}}.
   \end{align}
\end{lemma}
This simple relationship between the nonlocality and concurrence of a general two-qubit pure state has been mentioned in the literature~\cite{WJ01,WX02}.

Now, we consider the relationship between the nonlocality and entanglement when the general two-qubit states are considered. Verstraete and Wolf found that the nonlocality of a general two-qubit state is upper bounded by the function $2\sqrt{1 + C^{2}}$ of its concurrence $C$~\cite{FV02}. They found that the two-level density operator of the form
\begin{align}\label{eq:mtrixFV}
   \rho = \frac{1}{2}\left[\begin{array}{cccc}
      0 & 0 & 0 & 0 \\
      0 & 1-\alpha & C & 0 \\
      0 & C & 1 + \alpha & 0 \\
      0 & 0 & 0 & 0
   \end{array}
   \right]
\end{align}
can reach the upper bound. In fact, this is a necessary condition. Namely, any two-qubit state that can reach the upper bound must admit this form. However, they did not prove this fact in their paper. In the following, we rederive the upper bound and exploit the necessary and sufficient condition for reaching the upper bound. Finally, we prove the condition is equivalent to the example found by Verstraete and Wolf.

\begin{theorem}
   Suppose $\rho$ is an arbitrary density operator of a two-qubit system with concurrence being $C$. Then, the nonlocality of $\rho$ is upper bounded by
   \begin{equation}
      \mathcal{N}(\rho) \le 2\sqrt{1 + C^{2}},
   \end{equation}
   where the equality holds if and only if $\rho\in\mathcal{Q}$. The set $\mathcal{Q}$ of density operators on $\mathcal{H}_2^{\otimes 2}$ is defined as follows:
   \begin{align}
       \mathcal{Q} \equiv  \{(U_{A}\otimes U_{B})(p|\varphi_{1}\rangle \langle \varphi_{1}| + (1-p) |\varphi_{2}\rangle \langle \varphi_{2}| )(U_{A}\otimes U_{B})^{\dagger} \}  \nonumber
   \end{align}
   where $U_{A}$ and $U_{B}$ are arbitrary unitaries on $\mathcal{H}_{2}$, $p\in [0,1]$, $|\varphi_{1}\rangle = \cos{\theta} |00\rangle + \sin{\theta} |11\rangle$, and $|\varphi_{2}\rangle = \sin{\theta} |00\rangle + \cos{\theta} |11\rangle$  for some $\theta\in[0,\pi)$.
\end{theorem}
\begin{proof}
   According to Lemma~\ref{lemma:decomposition4concurrence}, there is a decomposition $\rho = \sum_{k} p_{k} |\psi_{k}\rangle\langle \psi_{k}|$, where $p_{k} \ge 0$ and $\sum_{k}p_{k} = 1$, such that $C(\psi_{k}) = C$ for all states $|\psi_{k}\rangle$. Combining the convexity of nonlocality and Lemma~\ref{lemma:NLandCC2QBpure}, we have
   \begin{align}
      \mathcal{N}(\rho) & \le \sum_{k} p_{k} \mathcal{N}(\psi_{k}) \label{ineq:convexityofNL} \\
       & = \sum_{k} p_{k} 2\sqrt{1 + C(\psi_{k})^{2}} \nonumber \\
       &= 2\sqrt{1 + C^{2}}. \nonumber
   \end{align}

   The equality in Eq.~(\ref{ineq:convexityofNL}) holds if and only if the optimal CHSH operator of the state $\rho$ also acts as the optimal CHSH operator for the state $|\psi_{k}\rangle$ for all $k$. Namely, all pure states $|\psi_{k}\rangle$ should have the same optimal CHSH operator in order to achieve the equality. Through the discussion in Sec.~\ref{sec:2qubitSameOptimalCHSHOperator}, we find that at most two different two-qubit pure states could have the same optimal CHSH operator. According to Theorem~\ref{thm:pure2QBsameOCHSHOperator}, any two different two-qubit pure states have the same optimal CHSH operator if and only if they can be written as $U_{A}\otimes U_{B}(\cos{\theta} |00\rangle + \sin{\theta} |11\rangle)$ and $U_{A}\otimes U_{B}(\sin{\theta} |00\rangle + \cos{\theta} |11\rangle)$ for some local unitaries $U_{A}$ and $U_{B}$ on $\mathcal{H}_{2}$ and $\theta\in[0,\pi)$. Equivalently, the equality in Eq.~(\ref{ineq:convexityofNL}) holds if and only if the state $\rho \in \mathcal{Q}$.

   Therefore, we have
   \begin{equation}
      \mathcal{N}(\rho) \le 2\sqrt{1 + C^{2}}
   \end{equation}
   with equality if and only if $\rho \in \mathcal{Q}$.
\end{proof}
The matrix representation of the operator $\rho(p,\theta) = p|\varphi_{1}\rangle \langle \varphi_{1}| + (1-p) |\varphi_{2}\rangle \langle \varphi_{2}|$ can be written as
\begin{align}
   \rho(p,\theta) = \frac{1}{2} \left[ \begin{array}{cccc}
      1-\alpha & 0 & 0 & C \\
      0 & 0 & 0 & 0 \\
      0 & 0 & 0 & 0 \\
      C & 0 & 0 & 1-\alpha
   \end{array}
   \right],
\end{align}
where $\alpha = 1 + 4p\sin^2{\theta} - 2p - 2 \sin^2{\theta}$. It is trivial to get that the parameter $\alpha$ ranges from $-1$ to $1$ according to the values of $p$ and $\theta$. Note that the matrix $\rho(p,\theta)$ is equivalent to Verstraete and Wolf's example in Eq.~(\ref{eq:mtrixFV}) up to a local unitary operation $I\otimes X$. Therefore, up to local unitary operations, the example in Verstraete and Wolf's paper~\cite{FV02} is the complete set of two-qubit states for reaching the upper bound.

\section{Conclusion and discussion}
To conclude, we have systematically investigated the quantitative relationship between quantum nonlocality and entanglement of a general two-qubit system. We rederived a known upper bound of the nonlocality of a general two-qubit state, which depends on the state's entanglement. We proved the necessary and sufficient condition for reaching the upper bound, which is the main contribution of our work. The condition is a class of two-qubit states which turns out to be equivalent to the class of states provided in Verstraete and Wolf's paper~\cite{FV02}. The key for solving the problem was to investigate the different two-qubit pure states that have the same optimal CHSH operator. First, we found the condition that two different general two-qubit states can have the same optimal CHSH operator. Second, we concluded that at most two two-qubit pure states can have the same optimal CHSH operator and figured out the class of such two-qubit pairs.

The research outcomes have practical significance. More efficient quantum information processing protocols could be designed by considering the quantum resources of which the most nonlocality feature could be stimulated by the same nonlocality test setting. In one situation, we only know the amount of entanglement of a given two-qubit resource while no knowledge about the structure of the state is known. Thus, it is difficult to get the quantity of the nonlocality. By applying the relationship of the nonlocality and entanglement, one is able to know an upper bound for the nonlocality of the state without nonlocality test experiments and the class of states that can reach the upper bound.

\begin{acknowledgments}
    The authors are partially supported by Anhui Initiative in Quantum Information Technologies (Grant No. AHY150100).
\end{acknowledgments}


\end{document}